\documentclass[onecolumn,10pt,aps,nofootinbib,superscriptaddress,tightenlines]{scrartcl}

\usepackage{amsmath}
\usepackage{amssymb}  
\usepackage{amsthm}
\usepackage{amsfonts}
\usepackage{graphicx}
\usepackage{bbm}
\usepackage{url}
\usepackage{ upgreek }
\usepackage{dsfont}
\usepackage{authblk}
\usepackage{color}
\usepackage{caption}
\usepackage{subcaption}

\newcommand{\tr}{\textnormal{tr}}

\newcommand{\ket}[1]{| #1 \rangle}

\newcommand{\bra}[1]{\langle #1 |}

\newcommand{\proj}[2]{| #1 \rangle\!\langle #2 |}

\newcommand{\id}{\ensuremath{\mathds{1}}}

\def\beq{\begin{equation}}
\def\eeq{\end{equation}}
\def\bq{\begin{quote}}
\def\eq{\end{quote}}
\def\ben{\begin{enumerate}}
\def\een{\end{enumerate}}
\def\bit{\begin{itemize}}
\def\eit{\end{itemize}}

\def\ra{\rightarrow}

\def\lb{\left(}
\def\rb{\right)}
\def\lset{\lbrace}
\def\rset{\rbrace}
\def\lk{\left\langle}
\def\rk{\right\rangle}
\def\l|{\left|}
\def\r|{\right|}
\def\lbr{\left[}
\def\rbr{\right]}
\def\i{\textnormal{id}}
\def\one{\id}

\newcommand\C{\mathbbm{C}}

\newcommand\R{\mathbbm{R}}
\newcommand\N{\mathbbm{N}}
\newcommand\M{\mathcal{M}}
\newcommand\D{\mathcal{D}}

\newcommand{\Tm}{\mathcal{T}}

\newcommand{\Lm}{\mathcal{L}}

\newcommand{\liou}{\mathcal{L}}
\newcommand{\ketbra}[1]{|#1\rangle\langle#1|}

\theoremstyle{plain}
\newtheorem{thm}{Theorem}[section]
\newtheorem{lem}{Lemma}[section]
\newtheorem{cor}{Corollary}[section]

\theoremstyle{definition}
\newtheorem{defn}{Definition}[section]

\begin{document}
\title{{Relative Entropy Convergence for Depolarizing Channels}}
\author{\vspace{-0.15cm}Alexander M\"uller-Hermes\thanks{muellerh@ma.tum.de}~}
\author{Daniel Stilck Fran\c{c}a\thanks{dsfranca@mytum.de}~}
\author{Michael M.~Wolf\thanks{wolf@ma.tum.de}~}
\affil{\vspace{-0.15cm}\small{Department of Mathematics, Technische Universit\"at M\"unchen, 85748 Garching, Germany}}
\maketitle

\begin{abstract}
We study the convergence of states under continuous-time depolarizing channels with full rank fixed points in terms of the relative entropy. The optimal exponent of an upper bound on the relative entropy in this case is given by the log-Sobolev-1 constant. Our main result is the computation of this constant. As an application we use the log-Sobolev-1 constant of the depolarizing channels to improve the concavity inequality of the von-Neumann entropy. This result is compared to similar bounds obtained recently by Kim et al. and we show a version of Pinsker's inequality, 
which is optimal and tight if we fix the second argument of the
relative entropy. Finally, we consider the log-Sobolev-1 constant of tensor-powers of the completely depolarizing channel and use a quantum version of Shearer's inequality to prove a uniform lower bound.    
\end{abstract}
\section{Introduction}

Let $\M_d$ denote the set of complex $d\times d$-matrices, $\D_d\subset \M_d$ the set of quantum states, i.e. positive matrices with trace equal to $1$, and $\D_d^+$ the set of strictly positive states. 
The relative entropy (also called quantum Kullback-Leibler divergence) of $\rho,\sigma\in\D_d$ is defined as 
\begin{align}
 D(\rho\|\sigma) :=
\begin{cases} 
\tr[\rho(\log\rho-\log\sigma)], & \mbox{if }  \mbox{supp}(\rho)\subset\mbox{supp}(\sigma)
\\ +\infty, & \mbox{otherwise}
\end{cases}.
\label{equ:RelEnt}
\end{align}
The relative entropy defines a natural distance measure to study the convergence of Markovian time-evolutions. For some state $\sigma\in\D_d$ consider the generalized depolarizing Liouvillian $\Lm_\sigma:\M_d\ra\M_d$ defined as
\begin{align}
\Lm_\sigma\lb\rho\rb := \text{tr}\lbr\rho\rbr\sigma - \rho.
\label{equ:DepLiou}
\end{align} 
This Liouvillian generates the generalized depolarizing channel $T^\sigma_t:\M_d\ra\M_d$ with $T^\sigma_t(\rho) := e^{t\Lm_\sigma}(\rho) = (1-e^{-t})\text{tr}\lbr\rho\rbr\sigma + e^{-t}\rho$, where $t\in\R^+$ denotes a time parameter. As $T^{\sigma}_t(\rho)\ra\sigma$ for $t\ra\infty$ we can study the convergence speed of the depolarizing channel with a full rank fixed point $\sigma\in\D_d$ by determining the largest constant $\alpha\in\R^+$ such that
\begin{align}
D(T^\sigma_t(\rho)\|\sigma)\leq e^{-2\alpha t}D(\rho\|\sigma)
\label{equ:LSInequInt}
\end{align}
holds for any $\rho\in\D_d$ and any $t\in\R^+$. This constant is known as the logarithmic Sobolev-1 constant~\cite{Olkiewicz1999246,Kastoryanosob} of $\Lm_\sigma$, denoted by $\alpha_1\lb\Lm_\sigma\rb$. In the following we will compute this constant and then use it to derive an improvement on the concavity of von-Neumann entropy.

\section{Preliminaries and notation}\label{sec:Prelim}

Consider a primitive\footnote{A Liouvillian is primitive if, and only if, it has a unique full rank fixed point $\sigma$ and for any $\rho\in\D_d$ we have $e^{t\Lm}(\rho)\ra\sigma$ as $t\ra\infty$.} Liouvillian $\Lm$ with full rank fixed point $\sigma\in\D_d$ and denote by $T_t := e^{t\Lm}$ the quantum dynamical semigroup generated by $\Lm$. Consider the function $f(t):=D\left(T_t(\rho)\big{\|}\sigma\right)$ for some 
initial state $\rho\in\D_d$ and note that if 
\begin{align*}
\frac{df}{dt}\leq-2\alpha f
\end{align*}
holds for some $\alpha\in\R_+$, then it follows that $f(t)\leq e^{-2\alpha t}f(0)$. The time derivative of the relative entropy at $t=0$, also called the entropy production~\cite{spohn},
is given by:
\begin{align}
 \frac{d}{dt} D\left(T_t(\rho)\Big{\|}\sigma\right)\bigg{|}_{t=0} &=-\tr[\liou(\rho)(\log(\sigma)-\log(\rho))] 
 \label{equ:EntrProd}
\end{align}
as $\tr(\liou(\rho)) = 0$ for any $\rho\in\D_d$.
This motivates the following definition:
\begin{defn}[log-Sobolev-1 constant, \cite{Olkiewicz1999246,Kastoryanosob}]
 For a primitive Liouvillian $\liou:\M_d\to\M_d$ with full rank fixed point $\sigma\in\D_d$ we define its \textbf{log-Sobolev-1 constant} as
 \begin{align}
  \alpha_1(\liou) :=\sup\Big{\lset} \alpha\in\R :\tr[\liou(\rho)(\log(\sigma)-\log(\rho))] \geq 2\alpha D\left(\rho\|\sigma\right),\forall \rho\in\D^+_d \Big{\rset}
  \label{equ:alpha1Def}
 \end{align}
\label{defn:LS-1}
\end{defn}
For a primitive Liouvillian $\liou:\M_d\ra\M_d$ the preceding discussion shows that \eqref{equ:LSInequInt} holds for any $\alpha \leq \alpha_1(\liou)$. Furthermore, $\alpha_1(\Lm)$ is the optimal constant for which this inequality holds independent of $\rho\in\D_d$ (for states $\rho$ not of full rank this follows from a simple continuity argument). 

In the following we will need some functions defined as continuous extensions of quotients of relative entropies. We denote by $Q_\sigma:\D^+_d\ra\R$ the continuous extension of the function $\rho\mapsto \frac{D(\sigma\|\rho)}{D(\rho\|\sigma)}$ (see Appendix \ref{sec:ContRelQu}) given by 
\begin{align}
Q_\sigma(\rho):=\begin{cases}\frac{D(\sigma\|\rho)}{D(\rho\|\sigma)}, &\rho\neq \sigma \\
1, &\rho=\sigma\end{cases}.
\label{equ:ContExt1}
\end{align}
Note that for $x\in[0,1]$ and $y\in(0,1)$ the binary relative entropy is defined as
\begin{align}
D_2(x\| y) := x\log\left(\frac{x}{y}\right) + (1-x)\log\left(\frac{1-x}{1-y}\right) .
\end{align}
This is the classical relative entropy of the probability distributions $(x,1-x)$ and $(y,1-y)$. For $y\in(0,1)$ we denote by $q_y:\lb 0,1\rb\ra\R$ the continuous extension of $x\mapsto \frac{D_2(y\|x)}{D_2(x\|y)}$ given by 
\begin{align}
q_y(x):=\begin{cases}\frac{D_2(y\|x)}{D_2(x\|y)}, &x\neq y \\
1, &x=y\end{cases}.
\label{equ:ContExt2}
\end{align}


%
%

\section{Log-Sobolev-1 constant for the depolarizing Liouvillian}

Note that for the depolarizing Liouvillian $\Lm_\sigma$ with $\sigma\in\D^+_d$ as defined in \eqref{equ:DepLiou} we have 
\begin{align*}
\tr[\liou_\sigma(\rho)(\log(\sigma)-\log(\rho))]=D(\rho\|\sigma)+D(\sigma\|\rho). 
\end{align*}
Inserting this into Definition \ref{defn:LS-1} we can write
\begin{align}
\alpha_1(\liou_\sigma) = \inf\limits_{\rho\in\D_d^+}\frac{1}{2}\Big{(}1 + Q_\sigma(\rho)\Big{)}.
\label{equ:alpha1DepLiou}
\end{align}  

Our main result is the following theorem:

\begin{thm}\label{thm:a1dcfinal}
Let $\liou_\sigma:\M_d\to\M_d$ be the depolarizing Liouvillian with full rank fixed point $\sigma\in\D_d$ as defined in \eqref{equ:DepLiou}. Then we have
\begin{align*}
\alpha_1\lb\liou_\sigma\rb = \min_{x\in\lbr 0,1\rbr}\frac{1}{2}\lb 1 + q_{s_\text{min}(\sigma)}(x) \rb,
\end{align*}
where $s_\text{min}(\sigma)$ denotes the minimal eigenvalue of $\sigma$.
\end{thm}

\begin{figure}
\centering
\includegraphics[width=10cm]{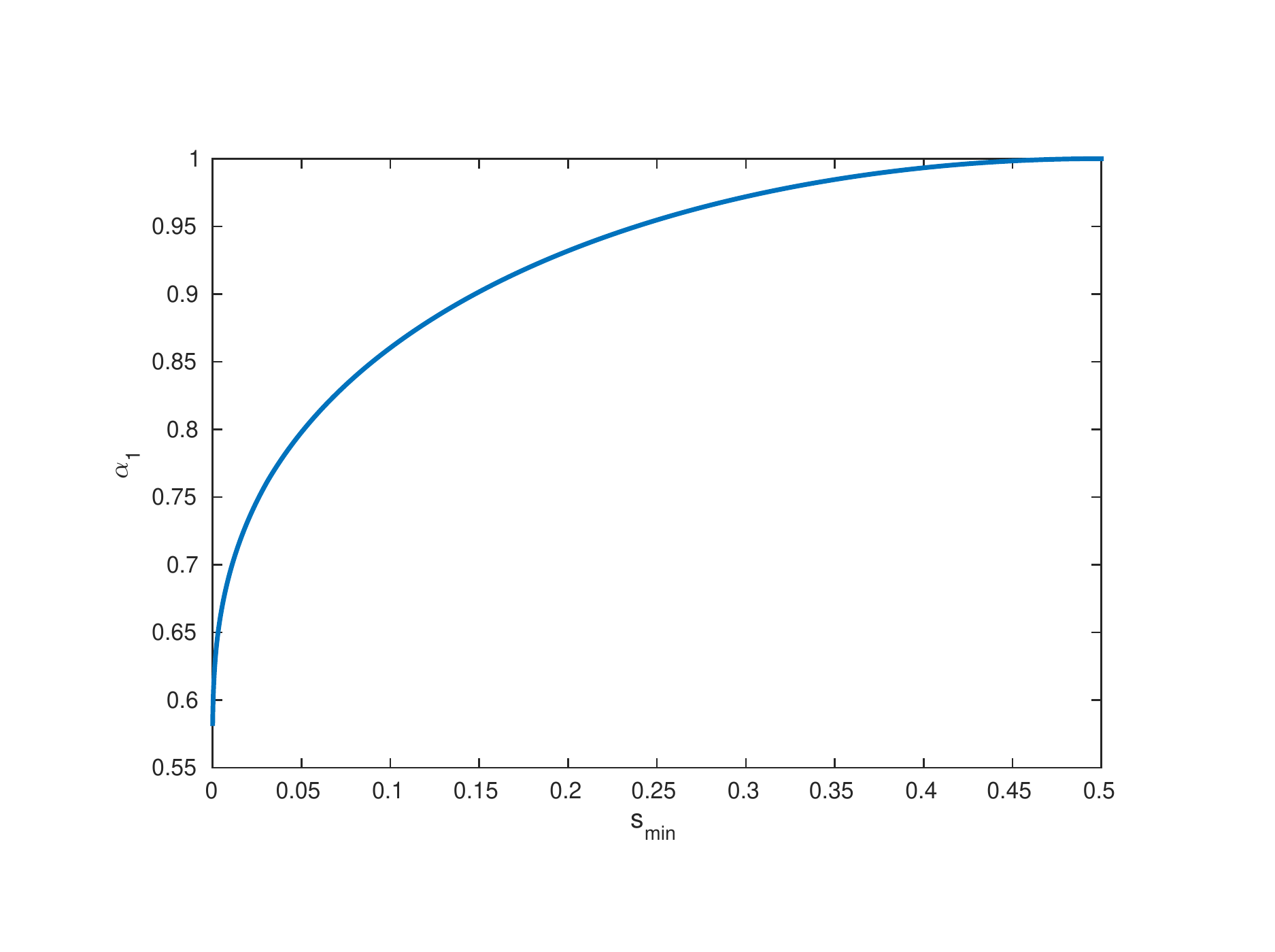}
\caption{$\alpha_1\lb\liou_\sigma\rb$ for $s_\text{min}(\sigma)\in\lbr 0,1\rbr$.}
\label{fig:alpha}
\end{figure} 

In Figure \ref{fig:alpha} the values of $\alpha_1\lb\liou_\sigma\rb$ depending on $s_\text{min}(\sigma)\in\lbr 0,1\rbr$ are plotted. Note that by Theorem \ref{thm:a1dcfinal} we have $\alpha_1\lb\liou_\sigma\rb\ra 1/2$ in the limit $s_\text{min}(\sigma)\ra 0$ (as $D_2(s_\text{min}(\sigma)\|x)\ra 0$ and $D_2(x\|s_\text{min}(\sigma))\ra \infty$ in this case). 

Before we state the proof of Theorem \ref{thm:a1dcfinal} we need to make a technical comment. By \eqref{equ:ContExt1} we have $Q_\sigma(\rho)\ra +\infty$ as $\rho\ra \partial \D_d$, i.e. as $\rho$ converges to a rank-deficient state. Therefore, the infimum in \eqref{equ:alpha1DepLiou} will be attained in a full rank state $\tilde{\rho}\in\D^+_{d}$ and we can restrict the optimization to the compact set $K_\sigma\subset \D_d$ (depending on $\sigma$) defined as 
\begin{align}
K_\sigma = \lset \rho\in\D^+_{d} : s_\text{min}(\rho)\geq s_\text{min}\lb\tilde{\rho}\rb -\epsilon\rset
\label{equ:Kset} 
\end{align}
for some fixed $\epsilon\in (0,s_\text{min}\lb\tilde{\rho}\rb)$ and where $s_\text{min}(\cdot)$ denotes the minimal eigenvalue. Note that the minimizing state $\tilde{\rho}$ is contained in the interior of $K_\sigma$. Now we have to solve the following optimization problem for fixed $\sigma\in\D^+_d$:
\begin{align}
\inf_{\rho\in\D^+_d} Q_\sigma(\rho) = \inf_{\rho\in K_\sigma} Q_\sigma(\rho).
\label{equ:BasicOpt}
\end{align} 

To prove Theorem \ref{thm:a1dcfinal} we will need the following lemma showing that the infimum in \eqref{equ:BasicOpt} is attained at states $\rho\in \D_d$ commuting with the fixed point $\sigma$. 

\begin{lem}\label{lem:infcommuting}

For any $\sigma\in\mathcal{D}^+_d$ we have
\begin{align*}
\inf_{\rho\in K_\sigma} Q_\sigma\lb\rho\rb = \inf_{\rho\in K_\sigma,\lbr\rho,\sigma\rbr = 0~}  Q_\sigma\lb\rho\rb
\end{align*}
where $Q_\sigma:D^+_d\ra\R$ denotes the continuous extension of $\rho\mapsto\frac{D(\sigma\|\rho)}{D(\rho\|\sigma)}$ (see \eqref{equ:ContExt1}).

\end{lem}

\begin{proof}

Consider the spectral decomposition $\sigma = \sum^d_{i=1}s_i \proj{v_i}{v_i}$ for $s\in\R_+^d$ and fix a vector $r\in\R_+^d$ which is not a permutation of $s$ and which fulfills $\min_i(r_i)\geq s_\text{min}\lb\tilde{\rho}\rb -\epsilon$ (see \eqref{equ:Kset}) and $\sum^d_{j=1} r_j = 1$. For some fixed orthonormal basis $\lset\ket{w_j}\rset_j$ consider $\rho := \sum^d_{j=1}r_j \proj{w_j}{w_j}\in K_\sigma$. Inserting $\rho$ into $Q_\sigma$ gives:
\begin{equation}\label{comeq1}
Q_\sigma(\rho) = \frac{D(\sigma\|\rho)}{D(\rho\|\sigma)} =\frac{-S(\sigma)-\tr[\sigma\log(\rho)]}{-S(\rho)-\tr[\rho\log(\sigma)]} = \frac{-S(\sigma)-\lk s,P\log(r)\rk}{-S(\rho)-\lk\log(s),Pr\rk} =: F(P)
\end{equation}
where we introduced $P\in\M_d$ given by $P_{ij} = |\lk v_i|w_j\rk|^2$ and $\log(s),\log(r)\in\R^d$ are defined as $(\log(s))_i = \log(s_i)$ and $(\log(r))_j = \log(r_j)$. Note that $P$ is a unistochastic matrix, i.e. a doubly stochastic matrix  whose entries are squares of absolute values of the entries of a unitary matrix. We will show that the minimum of $F$ over unistochastic matrices $P$ is attained at a permutation matrix. By definition of $P$ this shows that there exists a state $\rho'\in K_\sigma$ with spectrum $r$ and commuting with $\sigma$, which fulfills $Q_\sigma(\rho')\leq Q_\sigma(\rho)$.

As the set of unistochastic matrices is in general not convex~\cite{unistochasticproperties}, we want to consider the set of doubly stochastic matrices instead. By Birkhoff's theorem~\cite[Theorem II.2.3]{bhatia1997matrix} we can write any doubly stochastic $D\in\M_d$ as $D=\sum^k_{i=1} \lambda_i P_i$ for some $k\in\N$, numbers $\lambda_i\in\lbr 0,1\rbr$ with $\sum^k_{i=1}\lambda_i = 1$ and permutation matrices $P_i$. Now we can write the denominator of $F(D)$ as
\begin{align*}
-S(\rho)-\lk\log(s),Dr\rk = \sum^k_{i=1}\lambda_i\lb -S(\rho)-\lk\log(s),P_ir\rk\rb = \sum^k_{i=1}\lambda_i D(\rho_i\|\sigma) >0, 
\end{align*}  
where $\rho_i$ is the state obtained by permuting the eigenvectors of $\rho$ with $P_i$. In the last step we used Klein's inequality~\cite[p. 511]{nielsen2000quantum} together with the fact that $\rho_i\neq \sigma$ for any $1\leq i\leq k$ as their spectra are different. The previous estimate shows that $F$ is also well-defined on doubly stochastic matrices. 

Any unistochastic matrix is also doubly stochastic and we have
\begin{align*}
\inf\Big{\lset} F(P) : P\in\M_{d}\text{ doubly stochastic}\Big{\rset}\leq \inf\Big{\lset} F(P) : P\in\M_{d}\text{ unistochastic}\Big{\rset}.
\end{align*}
Note that $S(\sigma)$ and $S(\rho)$ in \eqref{comeq1} only depend on $s\in\R_+^d$ and $r\in\R_+^d$ and thus the numerator and the denominator of $F$ are positive affine functions in $P$. This shows that $F$ is a quasi-linear function~\cite[p. 91]{boyd2004convex}
 on the set of doubly stochastic matrices. It can be shown (see \cite{boyd2004convex}) that the minimum of such a function over a compact and convex set is always attained in an extremal point of the set. By Birkhoff's theorem~\cite[Theorem II.2.3]{bhatia1997matrix} the extremal points of the compact and convex set of doubly stochastic matrices are the permutation matrices. As these are also unistochastic matrices we have
\begin{align*}
\inf\Big{\lset} F(P) : P\in\M_{d}\text{ unistochastic}\Big{\rset} = \inf\Big{\lset} F(P) : P\in\M_{d}\text{ permutation matrix}\Big{\rset}.
\end{align*}
This finishes the first part. 

To prove the lemma note that we have
\begin{align*}
\inf_{\rho\in K_\sigma} Q_\sigma(\rho) = Q_\sigma(\tilde{\rho})
\end{align*}
for some minimizing full rank state $\tilde{\rho}\in\D^+_d$. Now consider some sequence $\lb \rho_n\rb_{n\in\N}\in K^\N_\sigma$ with $\rho_n\ra \tilde{\rho}$ as $n\ra\infty$ and such that the spectra of the $\rho_n$ are no permutations of the spectrum of $\sigma$. By the first part of the proof we find a sequence $\lb \rho'_n\rb_{n\in\N}\in K^\N_\sigma$ commuting with $\sigma$, such that 
\begin{align*}
Q_\sigma(\tilde{\rho})\leq Q_\sigma(\rho'_n)\leq Q_\sigma(\rho_n)\ra Q_\sigma(\tilde{\rho})
\end{align*}
as $n\ra \infty$.Thus $Q_\sigma(\rho'_n)\ra Q_\sigma(\tilde{\rho})$ as $n\ra\infty$. On the compact set $K_\sigma$ the sequence $\lb\rho'_n\rb_n$ has a converging subsequence $\lb\rho'_{n_k}\rb_{k\in\N}$ with $\rho'_{n_k}\ra \rho'\in K_\sigma$ as $k\ra\infty$. By continuity of $Q_\sigma$ we have $Q_\sigma(\rho')=Q_\sigma(\tilde{\rho}) = \inf_{\rho\in K_\sigma} Q_\sigma(\rho)$ and by continuity of the commutator $\rho\mapsto \lbr\rho,\sigma\rbr$ we have $\lbr\rho',\sigma\rbr=0$.

\end{proof} 

With this lemma we can prove our main result:

%
\begin{proof}[Proof of Theorem \ref{thm:a1dcfinal}]
By Lemma \ref{lem:infcommuting} we may restrict the optimization in \eqref{equ:BasicOpt} to states which commute with $\sigma$. Thus, we can repeat the construction of the compact set $K_\sigma$ (see \ref{equ:Kset}) for a minimizer $\tilde{\rho}\in\D^+_d$ with $\lbr\tilde{\rho},\sigma\rbr=0$. By construction $\tilde{\rho}$ lies in the interior of $K_\sigma$, which will be important for the following argument involving Lagrange-multipliers. 

To find necessary conditions on the minimizers of \eqref{equ:BasicOpt} we abbreviate $C:=\inf_{\rho\in K_\sigma} Q_\sigma(\rho)$ and note that $C>0$. 
To see this, note that we may extend $Q_\sigma(\rho)$ continuously to $1$ at $\sigma$, so there exists $\delta>0$ s.t. for $\|\rho-\sigma\|_1\leq\delta$ we
have $Q_\sigma(\rho)\geq\frac{1}{2}$ and for $\rho$ s.t. $\|\rho-\sigma\|_1>\delta$ we have $Q_\sigma(\rho)\geq\frac{\delta^2}{2\log\lb s_{\min}\lb\sigma^{-1}\rb\rb}$ Using Pinsker's inequality and $D(\rho\|\sigma)\leq\log\lb s_{\min}\lb\sigma\rb\rb$. For any $\rho\in K_\sigma$ with $\lbr\rho,\sigma\rbr=0$ and $\rho\neq \sigma$ have  
\begin{align*}
\frac{D(\sigma\|\rho)}{D(\rho\|\sigma)}\geq C
\end{align*}
which is equivalent to
\begin{equation}\label{eq1simpliinf}
S(\sigma)\leq C S(\rho)+C \sum_{i=1}^d r_i\log(s_i)-\sum_{i=1}^{d}s_i\log(r_i).
\end{equation}
Here $\lset r_i\rset^d_{i=1}$ denote the eigenvalues of $\rho\in K_\sigma$ (see \ref{equ:Kset}) fulfilling $\lbr \rho,\sigma\rbr=0$ and $\lset s_i\rset^d_{i=1}$ the eigenvalues of $\sigma$. As $\tilde{\rho}$ is a minimizer of \eqref{equ:BasicOpt} and commutes with $\sigma$ its spectrum is a minimizer of the right-hand-side of \eqref{eq1simpliinf} minimized over the set $\mathcal{S} := \lset r\in\R^d : \min_{i}(r_i)\geq s_\text{min}(\tilde{\rho})-\epsilon\rset\subset \R^d$ with $\epsilon$ chosen in the construction of $K_\sigma$ (see \eqref{equ:Kset}). We will now compute necessary conditions on the spectrum of $\tilde{\rho}$ using the formalism of Lagrange-multipliers (note that by construction the spectrum of $\tilde{\rho}$ lies in the interior of $\mathcal{S}$).

Consider the Lagrange function $F:\mathcal{S}\times \R\to\R$ given by
\begin{align*}
F(r_1,\ldots,r_d,\lambda)=CS(\rho)+C\sum_{i=1}^d r_i\log(s_i)-\sum_{i=1}^{d}s_i\log(r_i)+\lambda\lb\sum_{i=1}^d r_i-1\rb.
\end{align*}
The gradient of $F$ is given by:
\begin{equation}
\lbr\nabla F(r_1,\ldots,r_d,\lambda)\rbr_j=\begin{cases}
                         C(-\log(r_j)-1+\log(s_j))-\frac{s_j}{r_j}+\lambda &\quad 1\leq j\leq d\\    
       \sum_{i=1}^d r_i-1 &\quad j=d+1 \ 
                        \end{cases}
\end{equation}
By the formalism of Lagrange-multipliers any minimizer $r=(r_1,\ldots ,r_d)$ of the right-hand-side of \eqref{eq1simpliinf} in the interior of $\mathcal{S}$ has to fulfill $\nabla F(r_1,\ldots,r_d,\lambda)=0$ for some $\lambda\in\R$. Summing up the first $d$ of these equations (where the $j$th equation is multiplied with $r_j$) implies
\begin{align*}
\lambda=1+C(1+D(\rho\|\sigma)).
\end{align*}
Inserting this back into the equations $\lbr\nabla F(r_1,\ldots,r_d,\lambda)\rbr_j=0$ and using $u_j=\frac{r_j}{s_j}$ we obtain
\begin{equation}\label{eq3simpliinf}
u_j(1+C D(\rho\|\sigma))-1=Cu_j\log(u_j)
\end{equation}
for $1\leq j\leq d$. For fixed $D(\rho\|\sigma)$ there are only two values for $u_j$ solving the equations \eqref{eq3simpliinf}, as an affine functions (the left-hand-side) can only intersect a strictly convex function (the right-hand-side) in at most two points. Thus, for a minimizer $\lset r_i\rset^d_{i=1}$ of the right-hand-side of \eqref{eq1simpliinf} in the interior of $\mathcal{S}$ there are constants $c_1,c_2\in\R^+$ such that for each $i\in\lset 1,\ldots, d\rset$ either $r_i = c_1 s_i$ or $r_i = c_2 s_i$ holds. 

We have obtained the following conditions on the spectrum of the minimizer $\tilde{\rho}\in K_\sigma$ (fulfilling $\lbr\tilde{\rho},\sigma\rbr=0$) of \eqref{equ:BasicOpt}: There exist constants $c_1,c_2\in\R^+$ a permutation $\nu\in S_d$ (where $S_d$ denotes the group of permutations on $\lset 1,\ldots, d\rset$) and some $0\leq n\leq d$ such that the spectrum $r\in\R^+_d$ of $\tilde{\rho}$ fulfills $r_i = c_1 s_i$ for any $0\leq i\leq n$ and $r_i=c_2 s_i$ for any $n+1\leq i\leq d$. Note that the cases $c_1 = c_2=1$, $n=0$ and $n=d$ all correspond to the case $\rho=\sigma$ where we have $Q_\sigma(\sigma)=1$. Thus, we can exclude the cases $n=0$ and $n=d$ as long as we optimize over $c_1=c_2 =1$. Furthermore, note that we can use the normalization of $\tilde{\rho}$, i.e. $c_1\sum^n_{i=1}s_i + c_2\sum^d_{i=n+1}s_i=1$ to eliminate $c_2$. 
Given a permutation $\nu\in S_d$ and $n\in\{1,\ldots,d\}$, we define $p(\nu,n)=\sum^n_{i=1}s_{\nu(i)}$.
Inserting the above conditions into \eqref{equ:BasicOpt} and setting $c_1=x$ and $0<n<d$ yields
\begin{align}
\inf_{\rho\in K_\sigma} Q_\sigma(\rho) &= \inf_{\nu\in S_d}\inf_{1\leq n < d}\inf_{x\in \lbr 0,p(\nu,n)^{-1}\rbr} q_{p(\nu,n)}\lb xp(\nu,n)\rb\\
&= \inf_{\nu\in S_d}\inf_{1\leq n < d}\inf_{x\in \lbr 0,1\rbr} q_{p(\nu,n)}(x)
\label{equ:blub12}
\end{align}
where $q_y:\lbr 0,1\rbr\ra \R$ denotes the continuous extension of $x\mapsto\frac{D_2(y\| x)}{D_2(x\| y)}$ (see \eqref{equ:ContExt2}). By Lemma \ref{lem:concaveFunc} in the appendix the function $y\mapsto q_y(x)$ is continuous and quasi-concave and hence the minimum over any convex and compact set is attained at the boundary. Thus, we have 
\begin{align*}
q_{s_{\text{min}}(\sigma)}(x)\geq \inf_{\nu\in S_d}\inf_{1\leq n\leq d}q_{p(\nu,n)}(x)\geq \inf_{y\in \lbr s_{\text{min}}(\sigma), 1-s_\text{min}(\sigma)\rbr}q_{y}(x) = q_{s_{\text{min}}(\sigma)}(x)
\end{align*}
using $q_{1-s_\text{min}(\sigma)}(x)=q_{s_\text{min}(\sigma)}(x)$ for any $x\in\lbr 0,1\rbr$. Inserting this into \eqref{equ:blub12} leads to
\begin{align*}
\inf_{\rho\in K_\sigma}Q_\sigma(\rho) = \inf_{x\in\lbr 0, 1\rbr} q_{s_{\text{min}}(\sigma)}(x). 
\end{align*}

%
\end{proof}

Lemma \ref{lem:infcommuting} implies that the log-Sobolev-1 constant of the depolarizing channels coincides with the classical one of the random walk on the complete graph with $d$ vertices and distribution given by the spectrum of $\sigma$. This constant  has been shown to imply other inequalities, such as in \cite[Proposition 3.13]{Sammer05aspectsof}. Using this result, Theorem \ref{thm:a1dcfinal} implies a refined transportation inequality on graphs.

Using the correspondence with the classical log-Sobolev-1 constant of a random walk on the complete graph, we may apply \cite[Example 3.10]{bobkov2006modified}, which proves:

\begin{cor}
Let $\liou_\sigma:\M_d\to\M_d$ be the depolarizing Liouvillian with full rank fixed point $\sigma\in\D_d$ as defined in \eqref{equ:DepLiou}. Then we have
\begin{align*}
\alpha_1(\liou_{\sigma})\geq\frac{1}{2}+\sqrt{s_{\min}\lb\sigma\rb(1-s_{\min}\lb\sigma\rb)}
\end{align*}
with equality iff $s_{\min}\lb\sigma\rb=\frac{1}{2}$. Again $s_\text{min}(\sigma)$ denotes the minimal eigenvalue of $\sigma$.
\end{cor}

\section{Application: Improved concavity of von-Neumann entropy}

It is a well-known fact that the von-Neumann entropy $S(\rho)=-\text{tr}\lbr \rho\log(\rho)\rbr$ is concave in $\rho$. Using Theorem \ref{thm:a1dcfinal} we can improve the concavity inequality:

\begin{thm}[Improved concavity of the von-Neumann entropy]\label{thm:improvedconcav}
For $\rho,\sigma\in\D_d$ and $q\in[0,1]$ we have
\begin{align*}
S((1-q)\sigma+q\rho)-(1-q)S(\sigma)-q S(\rho)\geq \\ \max\begin{cases} q(1-q^{c(\sigma)})D(\rho\|\sigma)\\
(1-q)(1-(1-q)^{c(\rho)})D(\sigma\|\rho)\end{cases},
\end{align*}
with 
\begin{align*}
c(\sigma)=\min_{x\in[0,1]}\frac{D_2(s_{\min}(\sigma)\|x)}{D_2(x\|s_{\min}(\sigma))}
\end{align*}
and $c(\rho)$ defined in the same way.

\end{thm}

 Note that this bound becomes trivial if both $\sigma$ and $\rho$ are not of full rank (as we have $c(\rho)=c(\sigma)=0$ in this case). However, as long as $D(\rho\|\sigma)$ or $D(\sigma\|\rho)<\infty$, we may still get a bound by restricting both density matrices to the support of $\sigma$ or $\rho$, respectively.

\begin{proof}
 Note that for the Liouvillian $\liou := -\log(q)\liou_\sigma$ we have:
 \begin{align*}
e^{\liou}(\rho)=q\rho+(1-q)\sigma.
 \end{align*}
By Theorem \ref{thm:a1dcfinal} and \eqref{equ:LSInequInt} we have
\begin{equation}\label{eq1improvedconcav}
D\lb e^{\liou}(\rho)\|\sigma\rb\leq e^{(1+c(\sigma))\log(q)}D\lb\rho\|\sigma\rb
\end{equation}
Rearranging and expanding the terms in \eqref{eq1improvedconcav} we get
\begin{align*}
S(q\rho + (1-q)\sigma)&\geq (1-q)S(\sigma)-q\text{tr}\lbr\rho\log(\sigma)\rbr + q^{1+c(\sigma)}D(\rho\|\sigma) \\
&=(1-q)S(\sigma)+q S(\rho)+q(1-q^{c(\sigma)})D(\rho\|\sigma). 
\end{align*}
Interchanging the roles of $\rho$ and $\sigma$ in the above proof gives the second case under the maximum.

\end{proof}

In \cite{kim2014bounds} another improvement on the concavity of the von-Neumann entropy is shown:
\begin{align}
S((1-q)\sigma+q\rho)-(1-q)S(\sigma)-q S(\rho)&\geq\frac{q(1-q)}{(1-2q)^2}\max\begin{cases} D(\rho_{\text{avg}}\|\rho_\text{rev})\\ D(\rho_{\text{rev}}\|\rho_\text{avg}) \end{cases} \label{equ:KimBound1} \\
&\geq \frac{1}{2}q(1-q)\|\rho - \sigma\|^2_1 \label{equ:KimBound2}
\end{align}
where $\rho_\text{avg} = (1-q)\sigma+q\rho$ and $\rho_\text{rev} = (1-q)\rho+q\sigma$. Note that this bound is valid for all states $\rho,\sigma\in\D_d$ while our bound in Theorem \ref{thm:improvedconcav} 
becomes trivial unless the support $\rho$ is contained in the support of $\sigma$ or the other way around. We will therefore consider only full rank states in the following analysis.

By simple numerical experiments our bound from Theorem \ref{thm:improvedconcav} seems to be worse than \eqref{equ:KimBound1}. 
However, one can argue that \eqref{equ:KimBound1} is not much simpler than the left-hand-side itself. 
In particular the dependence on $\rho$ and $\sigma$ is only implicit via the relative entropy between $\rho_\text{avg}$ and $\rho_\text{rev}$. 
Our bound from Theorem \ref{thm:improvedconcav} depends on some spectral data (in terms of the smallest eigenvalues of $\rho$ or $\sigma$), 
but whenever this is given, we have a bound for any $q\in\lbr 0,1\rbr$ in terms of the relative entropies of $\rho$ and $\sigma$.

\begin{figure}
\centering
\begin{minipage}{.5\textwidth}
  \centering
  \includegraphics[width=6.5cm]{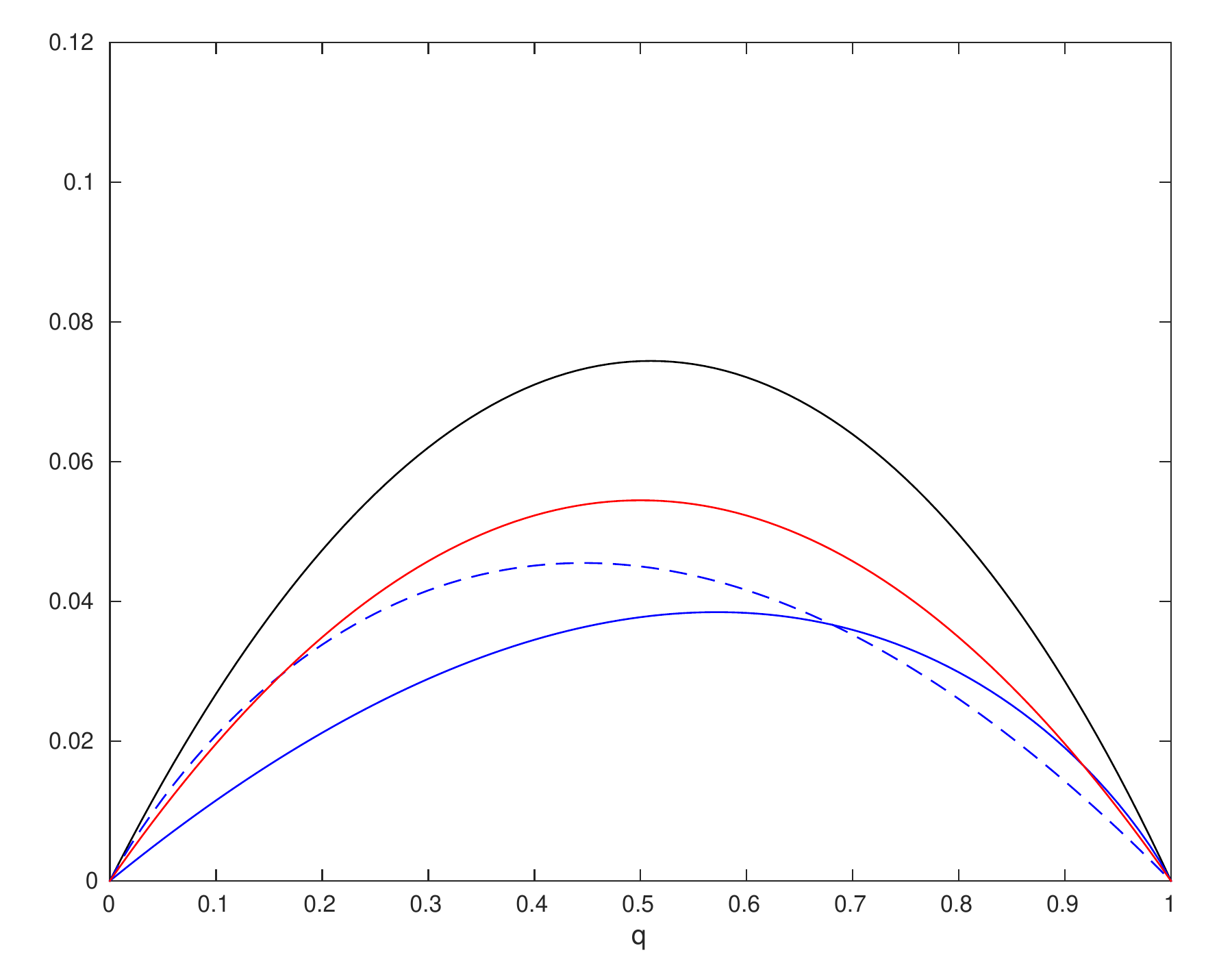}
  \captionof*{figure}{(a)}
\end{minipage}%
\begin{minipage}{.5\textwidth}
  \centering
  \includegraphics[width=6.5cm]{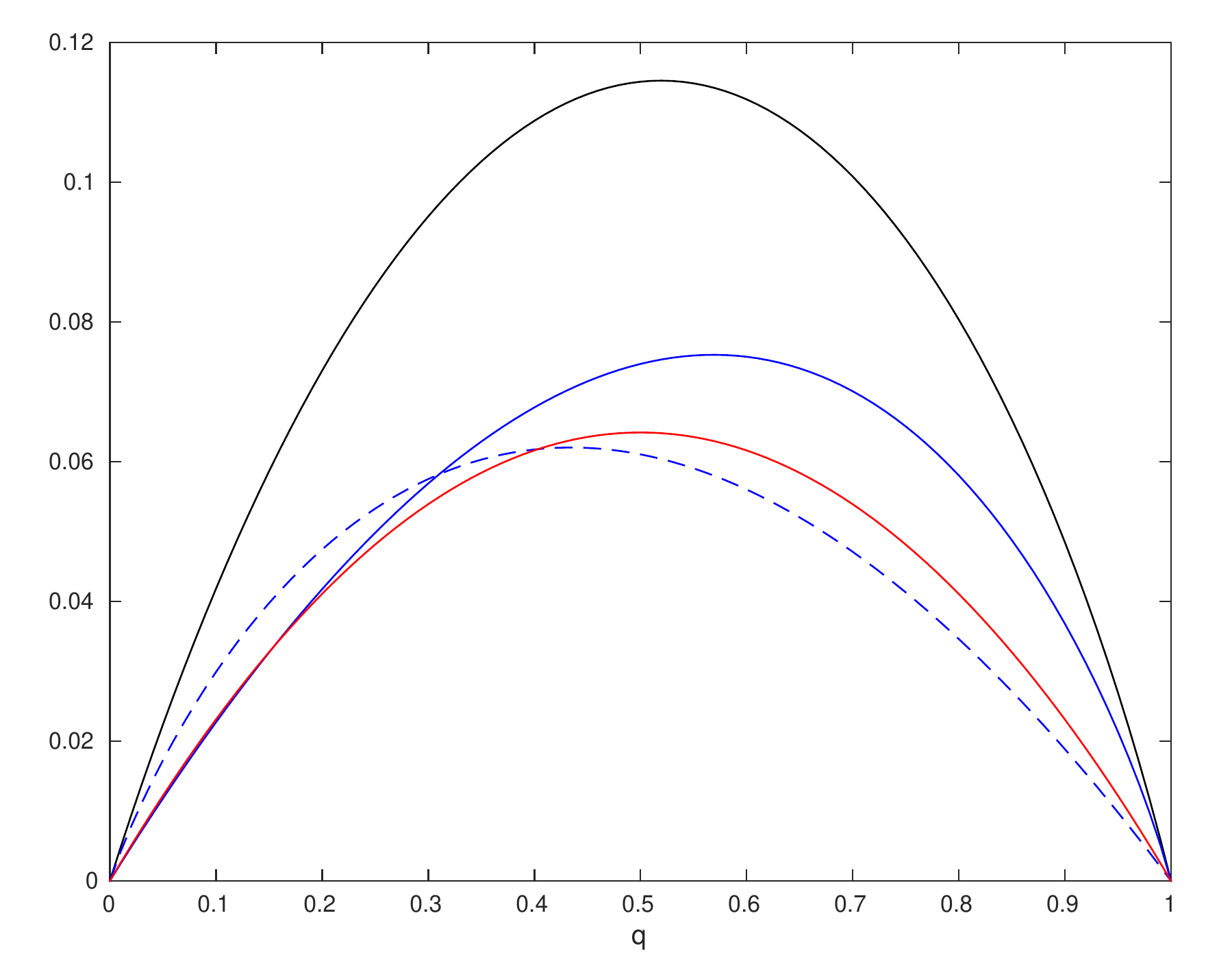}
  \captionof*{figure}{(b)}
\end{minipage}
\caption{Comparison of bound \eqref{equ:KimBound2} (red) and the bound from Theorem \ref{thm:improvedconcav} (blue, where both choices of the ordering of $\rho$ and $\sigma$ are plotted) and the exact value $S((1-q)\sigma+q\rho)-(1-q)S(\sigma)-q S(\rho)$ (black) two pairs of randomly generated $10\times 10$ quantum states and $q\in\lbr 0,1\rbr$.}
\label{fig:Comp}
\end{figure} 

Again we can do simple numerical experiments to compare the bounds \eqref{equ:KimBound2} and Theorem \ref{thm:improvedconcav}. 
Recall that our bound is given in terms of the relative entropy and \eqref{equ:KimBound2} in terms of the trace norm.
In Figure \ref{fig:Comp} the bounds are compared for randomly generated quantum states in dimension $d=10$. 
These plots show that the bounds are \emph{not} comparable and depending of the choice of the states the bound from Theorem \ref{thm:improvedconcav} will perform better than \eqref{equ:KimBound2} or vice versa. Note that for $q$ close to $0$ or $1$ our bound seems to perform better in both Figures. This is to be expected as $\alpha_1\lb\Lm_\sigma\rb$ is defined as the optimal constant $\alpha$ bounding the entropy production \eqref{equ:EntrProd} (in $t=0$) by $-2\alpha D(\rho\|\sigma)$. Therefore, Theorem \ref{thm:improvedconcav} should be the optimal bound (in terms of relative entropy) for $q$ near $0$ or $1$.

Note that by applying Pinsker's inequality:
\begin{align}
D\lb\rho\|\sigma\rb\geq\frac{1}{2}\|\rho-\sigma\|_1^2 
\label{equ:Pinsker}
\end{align} 
for states $\rho,\sigma\in\D_d$ to our bound from Theorem \ref{thm:improvedconcav} we can obtain an improvement on the 
concavity inequality in terms of the trace-distance similar to $\eqref{equ:KimBound2}$. Unfortunately a simple computation shows that the resulting trace-norm bound is always worse than $\eqref{equ:KimBound2}$. In the next section we will show that Pinsker's inequality can be improved in the case where the second argument in the relative entropy is fixed (which is the case in the bound from Theorem \ref{thm:improvedconcav}). This will lead to an additional improvement of the trace-norm bound obtained from Theorem \ref{thm:improvedconcav}, such that in some (but only very few) cases the bound becomes better than \eqref{equ:KimBound2}.

%
%

\section{State-Dependent Optimal Pinsker's Inequality}

Pinsker's inequality \eqref{equ:Pinsker} can be applied to the bound in Theorem \ref{thm:improvedconcav} to get an improvement of the concavity in terms of the trace distance of the two density matrices. 
It can also be applied to \eqref{equ:LSInequInt} to get a mixing time bound~\cite{Kastoryanosob} for the depolarizing channel. 
Note that in both of these cases the second argument of the relative entropy is fixed. Other improvements have been considered in the literature~\cite{audenaert2005continuity}, but here we will improve Pinsker's inequality in 
terms of the second argument of the relative entropy. 
More specifically we compute the optimal constant $C\lb\sigma\rb$ (depending on $\sigma$) such that $D\lb\rho\|\sigma\rb\geq C\lb\sigma\rb\|\rho-\sigma\|_1^2$ holds when $\sigma$ has full rank.

We will follow a strategy similar to the one pursued in \cite{ordentlich} in proving this, where the analogous problem was
considered for classical probability distributions. 
For a state $\rho\in\D_d$ let $s(\rho) = (s_1(\rho),\ldots,s_d(\rho))$ denote its vector of eigenvalues decreasingly ordered.
 \begin{lem}\label{inffixdis}
 Let $\sigma\in\D_d^+$ and for $A\subseteq\lset 1,\ldots ,d\rset$ define $P_{\sigma}\lb A\rb =\sum\limits_{i\in A}s_i(\sigma)$. Then we have for $\epsilon>0$:
 \begin{align*}
 \min\limits_{\rho:\|\rho-\sigma\|_{1}\geq\epsilon}D\lb\rho\|\sigma\rb =\min\limits_{A\subseteq\lset 1,\ldots ,d\rset}D_2\lb P_{\sigma}\lb A\rb +\epsilon\|P_{\sigma}\lb A\rb \rb 
 \end{align*}

 \end{lem}
\begin{proof}
Let $\rho\in\D_d$ be such that $\|\rho-\sigma\|_1=\delta$, with $\delta\geq\epsilon$. By Lidskii's theorem~\cite[Corollary III.4.2]{bhatia1997matrix}, we have:
\begin{align*}
s\lb\rho-\sigma\rb =s\lb\sigma\rb -Ls\lb\rho\rb , 
\end{align*}
where $L$ is a doubly stochastic matrix. Define $\rho'$ to be the state which has eigenvalues $Ls\lb\rho\rb $ and commutes with $\sigma$. Then we have:
\begin{align*}
\|\rho-\sigma\|_1=\|\rho'-\sigma\|_1 
\end{align*}
By the operational interpretation for the $1-$norm~\cite[Theorem 9.1]{nielsen2000quantum} there exist hermitian projections $Q,Q'\in\M_n$ such that
\begin{align}
2\tr[Q\lb\rho-\sigma\rb ] = \|\rho-\sigma\|_1 = \|\rho'-\sigma\|_1 =2\tr[Q'\lb\rho'-\sigma\rb ] .\label{equ:Norm2}
\end{align}
Now define the quantum channel $T:\M_d\to\M_2$ given by:
\begin{align*}
T\lb\rho\rb =\tr[Q\rho]\ketbra{0}+\tr[\lb\one-Q\rb \rho]\ketbra{1}.
\end{align*}
where $\ket{0},\ket{1}$ is an orthonormal basis of $\C^2$.
By the data processing inequality we have:
\begin{equation}\label{datapink}
 D\lb\rho\|\sigma\rb \geq D\lb T\lb\rho\rb \|T\lb\sigma\rb \rb 
\end{equation}
It is easy to see that the image of $Q'$ must be spanned by eigenvectors of $\sigma$.
Thus, we may associate a subset $A\subseteq\lset 1,\ldots ,d\rset$ to the projector $Q'$ indicating the eigenvectors of $\sigma$ spanning this subspace. 
Using \eqref{equ:Norm2} and the assumption that $\|\rho-\sigma\|_1 = \delta$ we have:
\begin{align*}
 \tr[Q'\rho' ]=P_{\sigma}\lb A\rb +\frac{\delta}{2}.
\end{align*}
Also observe that 
\begin{align*}
D\lb T\lb\rho\rb \|T\lb\sigma\rb \rb =D_2\lb P_{\sigma}\lb A\rb +\frac{\delta}{2}\Big{\|}P_{\sigma}\lb A\rb \rb \geq D_2\lb P_{\sigma}\lb A\rb +\frac{\epsilon}{2}\Big{\|}P_{\sigma}\lb A\rb \rb
\end{align*}
as the binary relative entropy is convex and $\delta\geq \epsilon$ was assumed. With \eqref{datapink} we have:
\begin{equation}\label{lowerpinsk}
 \min\limits_{\rho:\|\rho-\sigma\|_{1}\geq\epsilon}D\lb\rho\|\sigma\rb \geq\min\limits_{A\subseteq\lset 1,\ldots ,d\rset}D_2\lb P_{\sigma}\lb A\rb +\frac{\epsilon}{2}\Big{\|}P_{\sigma}\lb A\rb \rb 
\end{equation}
Now given any $A\subseteq\lset 1,\ldots ,d\rset$ such that $P_{\sigma}\lb A\rb +\frac{\epsilon}{2}<1$ (otherwise $D_2\lb P_{\sigma}\lb A\rb +\epsilon\|P_{\sigma}\lb A\rb \rb =+\infty$) , define a state $\tau\in\D_d$ which commutes with $\sigma$
and has spectrum:
\begin{align*}
s_i\lb\tau\rb =\begin{cases}
                         \frac{\lb P_{\sigma}\lb A\rb +\epsilon/2\rb s_i(\sigma)}{P_{\sigma}\lb A\rb } &\quad \text{ for } i\in A\\    
       \frac{\lb1-P_{\sigma}\lb A\rb -\epsilon/2\rb s_i(\sigma)}{1-P_{\sigma}\lb A\rb } &\quad \text{ else.} \ 
                        \end{cases}
\end{align*}
Note that $\|\sigma-\tau\|_1=\epsilon$ and $D\lb\tau\|\sigma\rb =D_2\lb P_{\sigma}\lb A\rb +\frac{\epsilon}{2}\|P_\sigma\lb A\rb \rb $, i.e.
the lower bound in \eqref{lowerpinsk} is attained.

\end{proof}
We define the function $\phi:[0,\frac{1}{2}]\to\R$ as 
\begin{align}
\phi\lb p\rb =\frac{1}{1-2p}\log\lb\frac{1-p}{p}\rb 
\label{equ:phi}
\end{align}
extended continuously by $\phi\lb\frac{1}{2}\rb =2$. Furthermore for any $\sigma\in\D_d$ we define 
\begin{align}
\pi\lb\sigma\rb =\max\limits_{A\subseteq\lset 1,\ldots ,d\rset}\min\left\{\frac{1}{2},\sum\limits_{i\in A}s_i(\sigma)\right\}.
\label{equ:pi}
\end{align}
With essentially the same proof as given in \cite{ordentlich} for the classical case we obtain the following improvement on Pinsker's inequality:

\begin{thm}[State-dependent Pinsker's Inequality]
For $\sigma,\rho\in\D_d$ we have:
\begin{align}
D\lb\rho\|\sigma\rb \geq\frac{\phi\lb\pi\lb\sigma\rb \rb }{4}\|\rho-\sigma\|_1^2 
\label{equ:ImproPinsker}
\end{align}
with $\phi$ as in \eqref{equ:phi} and $\pi\lb\sigma\rb$ as in \eqref{equ:pi}. Moreover, this inequality is tight.
\label{thm:ImproPinsker}
\end{thm}

\begin{proof}
For convenience set $\|\rho-\sigma\|_1 = \delta$. Then we have
 \begin{align}
D\lb\rho\|\sigma\rb \geq\min\limits_{\rho':\|\rho'-\sigma\|_1\geq\delta}D\lb\rho'\|\sigma\rb = \min\limits_{A\subseteq\lset 1,\ldots ,d\rset}D_2\lb P_{\sigma}\lb A\rb +\frac{\delta}{2}\|P_{\sigma}\lb A\rb \rb 
\label{ineq1pinsk}
 \end{align}
using Theorem \ref{inffixdis}.
By \cite[Proposition 2.2]{ordentlich} for $p\in[0,\frac{1}{2}]$ and $\epsilon\geq0$ we have
\begin{align*}
D_2\lb p+\epsilon\|p\rb \leq D_2\lb1-p+\epsilon\|1-p\rb 
\end{align*}
so we may assume $P_{\sigma}\lb A\rb \leq\frac{1}{2}$ in \eqref{ineq1pinsk}. In \cite[Theorem 1]{hoeffding1963probability} it is shown that for $p\in[0,\frac{1}{2}]$ we have
 \begin{align}
 \inf\limits_{\epsilon\in(0,1-p]}\frac{D_2\lb p+\epsilon\|p\rb }{\epsilon^2}=\phi\lb p\rb  
 \label{equ:InfPinsk}
 \end{align}
which implies:
\begin{align*}
\min\limits_{A\subseteq\lset 1,\ldots ,d\rset}D_2\lb P_{\sigma}\lb A\rb +\frac{\delta}{2}\Big{\|}P_{\sigma}\lb A\rb \rb \geq\min\limits_{A\subseteq\lset 1,\ldots ,d\rset}\frac{\phi\lb P_\sigma\lb A\rb \rb }{4}\|\rho-\sigma\|_1^2. 
\end{align*}
By \cite[Proposition 2.4]{ordentlich} the function $\phi$ is strictly decreasing. Thus, we have
\begin{align*}
 \min\limits_{A\subseteq\lset 1,\ldots ,d\rset}\frac{\phi\lb P_{\sigma}\lb A\rb \rb }{4}=\frac{\phi\lb\pi\lb\sigma\rb \rb }{4}
\end{align*}
which, after combining the previous inequalities, finishes the proof of \eqref{equ:ImproPinsker}. To show that the inequality is tight, we may again follow the proof of \cite[Proposition 2.1]{ordentlich}.
Let $B\subseteq\lset 1,\ldots ,d\rset$ be the subset such that $\pi\lb\sigma\rb =P_{\sigma}\lb B\rb =:p$. Define a minimizing sequence $\{\epsilon_i\}_{i\in\N}$ with $\epsilon_i>0$ for the infimum (with respect to p) in \eqref{equ:InfPinsk}, i.e. such that
\begin{align*}
\lim\limits_{i\to\infty}\frac{D_2\lb p+\epsilon_i\|p\rb }{\epsilon_i^2}=\phi\lb p\rb.  
\end{align*}
Next define a sequence of states $\rho_i$ that commute with $\sigma$ and have spectrum:
\begin{align*}
s_j\lb\rho_i\rb =\begin{cases}
                        \frac{\lb p+\epsilon_i\rb s_i(\sigma)}{p} &\quad \text{for }j\in B\\    
       \frac{\lb1-p-\epsilon_i\rb s_i(\sigma)}{1-p} &\quad\text{else. }  
                        \end{cases}
\end{align*}
One can check that $\|\rho_i-\sigma\|_1=2\epsilon_i$ and $D\lb\rho_i\|\sigma\rb =D_2\lb p+\epsilon_i\|p\rb $, from which we get:
\begin{align*}
 \lim\limits_{i\to\infty}\frac{D\lb\rho_i\|\sigma\rb }{\|\rho_i-\sigma\|^2}=\frac{\phi\lb\pi\lb\sigma\rb \rb }{4}
\end{align*}

\end{proof}
In some cases the bound can be made more explicit, as illustrated in the next corollary:
\begin{cor}\label{explicitexp}
Let $\sigma,\rho\in\D_d$ be such that $\|\sigma\|_\infty\geq\frac{1}{2}$. Then:
\begin{align}
D\lb\rho\|\sigma\rb \geq\frac{\phi\lb1-\|\sigma\|_\infty\rb }{4}\|\rho-\sigma\|_1^2  
\end{align}
\end{cor}
\begin{proof}
 In this case it is clear that $\pi\lb\sigma\rb =1-\|\sigma\|_\infty$.
\end{proof}
Note that we have $\phi\lb x\rb \to+\infty$ for $x\to0$. Thus, there might be an arbitrary large improvement of \eqref{equ:ImproPinsker} compared to the usual Pinsker's inequality \eqref{equ:Pinsker}. This happens for instance in Corollary \ref{explicitexp} when $\|\sigma\|_\infty\ra 1$, i.e. when $\sigma$ converges to a pure state. 

By applying the improved inequality \eqref{equ:ImproPinsker} to Theorem \ref{thm:improvedconcav} we obtain for quantum states $\rho,\sigma\in\D_d$ and $q\in[0,1]$ 
\begin{align*}
S((1-q)\sigma+q\rho)-(1-q)S(\sigma)-q S(\rho)\geq \max\begin{cases} q(1-q^{c(\sigma)})\frac{\phi\lb\pi\lb\sigma\rb\rb}{4}||\rho-\sigma||^2\\
(1-q)(1-(1-q)^{c(\rho)})\frac{\phi\lb\pi\lb\rho\rb\rb}{4}||\rho-\sigma||^2\end{cases}
\end{align*}
with $\phi$ as in \eqref{equ:phi} and $\pi\lb\sigma\rb$ as in \eqref{equ:pi}.

Even using this refinement of Pinsker's inequality, some numerical experiments indicate that \eqref{equ:KimBound2} is stronger for randomly generated states. From Corollary \ref{explicitexp} we can expect our bound to perform well if $\sigma$ has a large eigenvalue and the smallest eigenvalue is as large as possible. Such states have spectrum of the form $\lb p,\frac{1-p}{d-1},\ldots,\frac{1-p}{d-1}\rb$. Indeed for $\sigma\in\D_5$ with spectrum $\lb 0.99,0.0025,0.0025,0.0025,0.0025\rb$ and $q<0.2$ our bound performs better than \eqref{equ:KimBound1} for randomly generated $\rho$. However, even in this case the improvement is \emph{not} significant.

Still we can expect that Theorem \ref{thm:ImproPinsker} will find more applications, for instance improving the mixing time bounds. Such bounds have been derived from log-Sobolev inequalities in \cite{Kastoryanosob}. The next theorem can be used to improve these results: 

\begin{thm}
Let $\liou:\M_d\to\M_d$ be a primitive Liouvillian with fixed point $\sigma$ that satisfies
\begin{align}
D\lb e^{t\liou}\rho\|\sigma\rb \leq e^{-2\alpha t}D\lb \rho\|\sigma\rb  
\end{align}
for some $\alpha>0$ and for all $\rho\in\D_d$ and $t\in\R^+$. Then we have
\begin{align}
 \|e^{t\liou}(\rho)-\sigma\|_1\leq2e^{-\alpha t}\sqrt{\frac{\log\lb s_{\min}(\sigma)\rb}{\phi\lb \pi\lb \sigma\rb \rb }}
\end{align}
with $\phi$ as in \eqref{equ:phi}, $\pi\lb\sigma\rb$ as in \eqref{equ:pi} and where $s_{\min}(\sigma)$ is the smallest eigenvalue of $\sigma$.
\end{thm}
\begin{proof}
This is a direct consequence of \eqref{equ:LSInequInt} and $\eqref{equ:ImproPinsker}$.
\end{proof}

\section{Tensor products of depolarizing channels}

For a Liouvillian $\Lm:\M_d\ra\M_d$ generating the channel $\Tm_t = e^{t\Lm}$ and any $n\in\N$ we denote by $\Lm^{(n)}:\M_{d^n}\ra\M_{d^n}$ the generator of the tensor-product semigroup $(T_t)^{\otimes n}$, i.e. $\Lm^{(n)} := \sum^n_{i=1} \i_d^{\otimes i-1}\otimes\liou\otimes \i_d^{\otimes (n-i)}$. 

Here we study $\alpha_1\lb\liou^n_\sigma\rb$ in the special case where $\sigma=\frac{\one_d}{d}$. For simplicity we denote the depolarizing Liouvillian onto $\sigma = \frac{\one_d}{d}$ by $\Lm_{d}:= \Lm_{\frac{\one_d}{d}}$ and by $T_t^{d} = e^{t\Lm^{d}}$ the generated semigroup. In the case $d=2$ it is known~\cite{Kastoryanosob} that $\alpha_1\lb\liou_{2}^{(n)}\rb=1$ for any $n\in\N$ .
It is, however, an open problem to determine this constant for any $d>2$ and any $n\geq 2$. We will now show the inequality $\alpha_1\lb\Lm^{(n)}_{d}\rb\geq\frac{1}{2}$ for any $d\geq 2$ and $n\geq 1$, which is the best possible
lower bound that is independent of the local dimension. 
Note that for $\sigma = \frac{\one_d}{d}$ inequality \eqref{equ:LSInequInt} for the channel $\lb\Tm^d_t\rb^{\otimes n}$ can be rewritten as the entropy production inequality:
\begin{align*}
S((T^d_t)^{\otimes n}(\rho)) \geq (1-e^{-t})n\log(d) + e^{-t} S(\rho) . 
\end{align*}

This inequality has been studied in \cite{benor} for the case where $d=2$, for wich, however, an incorrect proof was given. We will provide a proof of a more general statement, from which the claim $\alpha_1\lb\Lm^{(n)}_{\text{dep}}\rb\geq\frac{1}{2}$
readily follows by the previous discussion.
\begin{thm}
For any $\sigma,\rho\in \D_d$ (not necessarily full rank) we have
\begin{align*}
S((T^\sigma_t)^{\otimes n}(\rho)) \geq e^{-t} S(\rho) + (1-e^{-t})S(\sigma^{\otimes n}).
\end{align*}
\label{thm:BenorBound}
\end{thm}

For the proof we will need a special case of the quantum Shearer's inequality. We will denote by $\rho\in\D\lb \C^{d_1}\otimes \C^{d_2}\otimes \cdots\otimes \C^{d_n}\rb$ a multipartite density matrix (where the $d_i$ are the local dimensions of each tensor factor). Furthermore we write $S(i_1,i_2,\ldots ,i_k)_\rho$ for the entropy of the reduced density matrix $\rho$ on the tensor factors specified by the indices $i_1,i_2,\ldots, i_k$. 
Similarly we write 
\begin{align*}
S(i_1,\ldots ,i_k|j_1,\ldots,j_l)_\rho = S(i_1,\ldots,i_k,j_1,\ldots ,j_l)_\rho - S(j_1,\ldots,j_l)_\rho
\end{align*}
for a conditional entropy. The proof of the quantum version of Shearer's inequality is essentially the same as the proof given by Radhakrishnan and Llewellyn for the classical version (see~\cite{radhakrishnan}). For convenience we provide the full proof:

\begin{lem}[Quantum Shearer's inequality]
Consider $t\in\N$ and a family $\mathcal{F}\subset 2^{\lset 1,\ldots , n\rset}$ of subsets of $\lset 1,\ldots , n\rset$ such that each $i\in\lset 1,\ldots, n\rset$ is included in exactly $t$ elements of $\mathcal{F}$. Then for any $\rho\in\D\lb \C^{d_1}\otimes \C^{d_2}\otimes \cdots\otimes \C^{d_n}\rb$ we have 
\begin{align}
S(1, 2\ldots, n)_\rho\leq \frac{1}{t}\sum_{F\in \mathcal{F}} S(F)_\rho~.
\end{align}

\end{lem}
\begin{proof}
For $F\subset \lset 1,\ldots ,n\rset$ denote its elements by $\lb i_1,\ldots ,i_k\rb$, increasingly ordered. For any $\rho\in\D\lb \C^{d_1}\otimes \C^{d_2}\otimes \cdots\otimes \C^{d_n}\rb$ we have
\begin{align*}
\sum^{|F|}_{j=1} S(i_j|i_1,\ldots,i_{j-1})_\rho &= S(i_1)_\rho + S(i_2|i_1)_\rho + \cdots + S(i_{|F|}|i_1,i_2,\ldots ,i_{|F|-1})_\rho \\
&= S(i_1,i_2,\ldots ,i_{|F|})_\rho = S(F)_\rho
\end{align*}
where we used a telescopic sum trick. By strong subadditivity~\cite{lieb1973} conditioning decreases the entropy. This implies
\begin{align}
\sum^{|F|}_{j=1} S(i_j|1,2,\ldots,i_{j}-1)_\rho\leq \sum^{|F|}_{j=1} S(i_j|i_1,\ldots,i_{j-1})_\rho = S(F)_\rho.
\label{equ:Shearer1}
\end{align}
Now consider a family $\mathcal{F}\subset 2^{\lset 1,\ldots , n\rset}$ with the properties stated in the assumptions. Using \eqref{equ:Shearer1} for the first inequality gives:
\begin{align}
\sum_{F\in\mathcal{F}} S(F)_\rho &\geq \sum_{F\in\mathcal{F}} \sum^{|F|}_{j=1} S(i_j|1,2,\ldots ,i_{j}-1)_\rho \label{equ:Shearer14} \\
&= t\sum^n_{i=1} S(i|1,2,\ldots,i-1)_\rho = t S(1,2,\ldots ,n)_\rho. \label{equ:Shearer15}
\end{align} 
Here we used the assumption that each $i\in\lset 1,\ldots ,n\rset$ is contained in exactly $t$ elements of $\mathcal{F}$ and $\eqref{equ:Shearer1}$ in the special case of $F=\lset 1,\ldots ,n\rset$ for the final equality.
\end{proof}
Note that in the classical case Shearer's inequality is true under the weaker assumption that any $i\in\lset 1,\ldots,d\rset$ is contained in \emph{at least} $t$ elements of $\mathcal{F}$. However, as the quantum conditional entropy might be negative~\cite{PhysRevLett.79.5194} we have to use the stronger assumption to get the equality between \eqref{equ:Shearer14} and \eqref{equ:Shearer15} where an $\geq$ would be enough.

In the special case where $\mathcal{F}=\mathcal{F}_k := \lset F\subseteq \lset 1,\ldots,n\rset : |F|=k\rset$ denotes the family of $k$-element subsets of $\lset 1,\ldots,n\rset$ (i.e. every $i\in\lset 1,\ldots ,d\rset$ is contained in exactly $\binom{n-1}{k-1}=\frac{k}{n}\binom{n}{k}$ elements of $\mathcal{F}_k$) the quantum Shearer inequality gives
\begin{align}
\frac{k}{n}S(1,\ldots,n)\leq \frac{1}{\binom{n}{k}}\sum_{F\in\mathcal{F}_k} S(F).
\label{equ:Shearer2}
\end{align} 
This inequality was also proved in~\cite{junge2014cb}, but in a more complicated way and without mentioning the more general quantum Shearer's inequality. It is also used as a lemma (with wrong proof) in \cite{benor}, where the rest of the proof of their entropy production estimate is correct. The proof of Theorem \ref{thm:BenorBound} follows the same lines. For completeness we will include the full proof here:
\begin{proof}[Proof of Theorem \ref{thm:BenorBound}]
In the following we will abbreviate $p:=e^{-t}$. For a subset $F\subset \lset 1,\ldots , n\rset$ we denote by $\rho|_F$ the reduced density matrix on the tensor factors specified by $F$. Using this notation we can write
\begin{align*}
(T^\sigma_t)^{\otimes n}(\rho) = \sum^n_{k=0}\sum_{F\in\mathcal{F}_k}(1-p)^kp^{n-k} \lb\bigotimes_{l\in F}\sigma \otimes \rho|_{F^c}\rb
\end{align*} 
where $F^c = \lset 1,\ldots ,n\rset\setminus F$. Concavity of the von-Neumann entropy implies
\begin{align*}
S\lb (T^\sigma_t)^{\otimes n}(\rho)\rb &\geq \sum^n_{k=0}\sum_{F\in\mathcal{F}_k}(1-p)^kp^{n-k} \lb kS(\sigma) + S(F^c)_\rho\rb\\
&\geq (1-p)nS(\sigma) + \sum^n_{k=0} \binom{n}{n-k}\frac{n-k}{n}(1-p)^kp^{n-k}S(\rho)\\
&= (1-p)S(\sigma^{\otimes n}) + pS(\rho).
\end{align*}
Here we used the elementary identity $\sum^n_{k=0}\binom{n}{k}(1-p)^kp^{n-k}k = (1-p)n$ and \eqref{equ:Shearer2} for the $(n-k)$-element subsets $F^c$. 
\end{proof}

\section*{Acknowledgments}
We thank Ion Nechita for pointing out Shearer's inequality to us and David Reeb for interesting and helpful discussions. D.S.F acknowledges support from the graduate program TopMath of the Elite Network of Bavaria, the TopMath Graduate Center of TUM Graduate School at Technische Universit\"{a}t M\"{u}nchen and the Deutscher Akademischer Austauschdienst(DAAD). We also acknowledge financial support from the CHIST-ERA/BMBF project CQC (A.M.H.\ and M.M.W.).

\appendix

\section{Quasi-concavity of a quotient of relative entropies}

In this appendix we will prove the quasi-concavity of the function $y\mapsto q_y(x)$ for any $x\in\lb 0,1\rb$. As defined in \eqref{equ:ContExt2} the function $q_y:(0,1)\ra\R$ denotes the continuous extension of $x\mapsto\frac{D_2(y\| x)}{D_2(x\| y)}$. In the following we consider $f_x:\lbr 0,1\rbr\ra\R$ defined as $f_x(y)=q_y(x)$ for any $y\in(0,1)$ and with $f_x(0)=f_x(1)=1$. It can be checked easily that $f_x$ is continuous for any $x\in (0,1)$. We have the following Lemma:

\begin{lem}

For any $x\in (0,1)$ the function $f_x :\lbr 0,1\rbr\ra\R$ given by $f_x(y) = \frac{D_2(y\|x)}{D_2(x\|y)}$ for $y\notin\lset 0,x,1\rset$ and extended continuously by $f_x(x)=1$ and $f_x(0)=f_x(1)=0$ is quasi-concave.
\label{lem:concaveFunc}
\end{lem}

\begin{proof}

Note that without loss of generality we can assume $x\geq \frac{1}{2}$, as $f_x(y) = f_{1-x}(1-y)$. By continuity it is clear that there exists an $m_f\in (0,1)$ (we can exclude the boundary points since $f_x(x)>f_x(0)=f_x(1)$) such that $f_x(m_f)$ is the global maximum. By~\cite{boyd2004convex}[p. 99] it is sufficient to show that $f_x$ is unimodal, i.e. that $f_x$ is monotonically increasing on $\lbr 0,m_f\rb$ and monotonically decreasing on $\lb m_f,1\rbr$. We will use the method of L'Hospital type rules for monotonicity developed in \cite{pinelis2002hospital,pinelis2007non}.

For any $x\in (0,1)$ and $y\in(0,1)$ with $x\neq y$ we compute 
\begin{align*}
&\partial_y D_2\lb y\|x\rb = \log\lb\frac{y(1-x)}{x(1-y)}\rb \hspace{1.25cm} \partial_y D_2\lb x\| y\rb = \frac{y-x}{y(1-y)} \\
&\partial_y \log\lb\frac{y(1-x)}{x(1-y)}\rb = \frac{1}{y(1-y)} \hspace{1cm} \partial_y \frac{y-x}{y(1-y)} = \frac{y^2 + x -2yx}{(1-y)^2 y^2} 
\end{align*}
and define
\begin{align}
g_x(y) &= \frac{\partial_y D_2\lb y\|x\rb}{\partial_y D_2\lb x\| y\rb} = \frac{\log\lb\frac{x(1-y)}{y(1-x)}\rb y(1-y)}{x-y} \label{Der1}\\
h_x(y) &= \frac{\partial_y \log\lb\frac{x(1-y)}{y(1-x)}\rb}{\partial_y \frac{x-y}{y(1-y)}} = \frac{y(1-y)}{y^2 + x -2yx} \label{equ:Der2} 
\end{align}
where again $g_x$ is extended continuously by $g_x(0)=g_x(1)=0$ and $g_x(x)=1$. As $y\mapsto y^2 + x -2yx$ has no real zeros for $x\in(0,1)$ the rational function $h_x$ is continuously differentiable on $(0,1)$. A straightforward calculation reveals that for $x\geq \frac{1}{2}$ and on $(0,1)$ the derivative $h_x'$ only vanishes in 
\begin{align*}
m_h = \begin{cases} \frac{x-\sqrt{x(1-x)}}{2x-1} &\text{ for }x>\frac{1}{2} \\ ~\frac{1}{2} &\text{ for } x=\frac{1}{2}  \end{cases}.
\end{align*}
which has to be a maximum as $h_x(0)=h_x(1)=0$. By the lack of further points with vanishing derivative we have $h_x'(y)<0$ for any $y<m_h$ and also $h_x'(y)>0$ for any $y>m_h$. Note that $m_h\leq x$ for any $x\geq \frac{1}{2}$. 

Consider first the interval $(x,1)\subset (0,1)$. For $y\ra x$ we have $\log\lb\frac{x(1-y)}{y(1-x)}\rb\ra 0$  and $\frac{x-y}{y(1-y)}\ra 0$. Also it is clear that $y\mapsto \frac{x-y}{y(1-y)}$ does not change sign on the interval $(x,1)$. Therefore and by \eqref{equ:Der2} we see that the pair $g_x$ and $h_x$ satisfy the assumptions of \cite[Proposition 1.1.]{pinelis2002hospital} and as $h_x$ is decreasing we have that $g_x'(y)<0$ for any $y\in(x,1)$. We can use the same argument for the (possibly empty) interval $(m_h,x)$ where $h_x$ is decreasing as well and obtain $g_x'(y)<0$ for any $y\in (m_h,x)$. By continuity of $g_x$ in $x$ we see that $g_x$ is decreasing on $(m_h,1)$. 

Note that in the case where $x=\frac{1}{2}$ we can directly apply \cite[Proposition 1.1.]{pinelis2002hospital} to the remaining interval $(0,\frac{1}{2})$ where $h_{1/2}$ is increasing. This proves $g_{1/2}'(y)>0$ for any $y\in (0,\frac{1}{2})$. By continuity $m_g = \frac{1}{2}$ is the maximum point of $g_{1/2}$. For $x\neq \frac{1}{2}$, where the remaining interval is $(0,m_h)$ we apply the more general Proposition 2.1. in \cite{pinelis2007non}. It can be checked easily that the assumptions of this proposition are fulfilled for the pair $g_x$ and $h_x$. As for $y\in (0,m_h)$ we have $\frac{y-x}{y(1-y)}\frac{y^2 + x -2yx}{(1-y)^2y^2}<0$ and as $h_x$ is increasing the proposition shows that $g_x'(y)>0$ for any $y\in (0,m_g)$ and $g_x'(y)<0$ for any $y\in (m_g,m_h)$. Here $m_g\in (0,m_h)$ denotes the maximum point of $g_x$ (note that a maximum $m_g$ has to exist due to continuity and $g_x(0)=g_x(1)=0$).   

The previous argument shows that for any $x\geq \frac{1}{2}$ there exists a point $m_g\in (0,m_h]\subset (0,x]$ (we have $m_g = m_h = \frac{1}{2}$ for $x=\frac{1}{2}$) such that $g_x'(y)>0$ for $y\in (0,m_g)$ and $g_x'(y)<0$ for $y\in (m_g,1)\setminus\lset x\rset$. We can now repeat the above argument for the pair $f_x$ and $g_x$. This gives the existence of a point $m_f\in (0,m_g]$ such that $f_x'(y)>0$ for any $y\in (0,m_f)$ and $f_x'(y)<0$ for any $y\in (m_f,1)\setminus\lset x\rset$. By continuity in $x$ this shows that the function $f_x$ is unimodal and therefore quasi-concave.

\end{proof}

\section{Continuous extension of a quotient of relative entropies}\label{sec:ContRelQu}

In this section we show that the function $Q_\sigma:\D^+_d\ra\R$ as defined in \eqref{equ:ContExt1} is indeed continuous. As $Q_\sigma$ is clearly continuous in any point $\rho\neq \sigma$ we have to prove the following: 

\begin{lem}
For $\sigma\in\D_d^+$ and $X\in\M_d$ with $X=X^\dagger$, $\tr[X]=0$ and $X\not=0$ we have
\begin{align*}
\lim\limits_{\epsilon\to0}\frac{D\lb\sigma||\sigma+\epsilon X\rb}{D\lb\sigma+\epsilon X||\sigma\rb}=1 .
\end{align*}

\end{lem}
\begin{proof}
To show the claim we will expand the relative entropy in terms of $\epsilon$ up to second order.
Observe that for $\rho\in\D_d$ we have
\begin{align}
D\lb\rho\|\sigma\rb=\int\limits_0^{\infty}\tr\left[\rho\lb\lb\rho+t\rb^{-1}-\lb\sigma+t\rb^{-1}\rb\right]dt .
\label{equ:relativeEntInt}
\end{align}
In the following we assume $\epsilon>0$ to be small enough such that $\sigma+\epsilon X\in\D_d^+$.
To simplify the notation, we introduce $A(t):=\lb\sigma+t\rb^{-1}$ and $B(t):=\lb\sigma+\epsilon X+t\rb^{-1}$.
Applying the recursive relation 
\begin{align*}
B(t)=-\epsilon B(t)XA(t)+A(t),
\end{align*}
twice leads to
\begin{align*}
B(t)-A(t) &=-\epsilon B(t)XA(t)=
\epsilon^2 B(t)XA(t)XA(t)-\epsilon A(t)XA(t)\\
&=\epsilon^2A(t)XA(t)XA(t)-\epsilon A(t)XA(t)+\mathcal{O}\lb\epsilon^3\rb. 
\end{align*}
Inserting this into \eqref{equ:relativeEntInt} gives
\begin{align}
D\lb\sigma\|\sigma+\epsilon X\rb =\int\limits_0^{\infty}\tr\left[\epsilon\sigma A(t)XA(t)-\epsilon^2\sigma A(t)XA(t)XA(t)+\mathcal{O}(\epsilon^3)\right]dt
\label{equ:relInt1}
\end{align}
and
\begin{align}
D\lb\sigma+\epsilon X\|\sigma\rb &=\int\limits_0^{\infty}\tr\left[-\epsilon\sigma A(t)XA(t)+\epsilon^2\sigma A(t)XA(t)XA(t)-\epsilon^2XA(t)XA(t)+\mathcal{O}(\epsilon^3)\right]dt.
\label{equ:relInt2}
\end{align}
As $\lbr A(t),\sigma\rbr = 0$ we can diagonalize these operators in the same orthonormal basis $\lset\ket{i}\rset\subset \C^d$, which leads to
\begin{align}\label{firstorder}
\int\limits_0^{\infty}\tr\left[\sigma A(t)XA(t)\right]dt=\sum\limits_{i=1}^d \bra{i}X\ket{i}\int\limits_0^{\infty}\frac{s_i}{(s_i+t)^2}dt
=\sum\limits_{i=1}^d\bra{i}X\ket{i}=0
\end{align}
where $\lset s_i\rset^d_{i=1}$ denotes the spectrum of $\sigma$. Note that again by diagonalizing $\sigma$ and $A(t)$ in the same basis we have
\begin{align}
\int\limits_0^{\infty}&\tr\left[\lb2\sigma A(t)-\one\rb\lb XA(t)XA(t)\rb\right]dt\label{equ:blaInt} \\
&=\sum\limits_{i,j=1}^{d}|\bra{i}X\ket{j}|^2\int\limits_0^{\infty}\frac{2s_i}{(s_i+t)^2(s_j+t)}-\frac{1}{(s_i+t)(s_j+t)}dt\nonumber\\
\label{equ:blup46}
&=\sum\limits_{i,j=1}^{d}\frac{|\bra{i}X\ket{j}|^2}{(s_i-s_j)^2}\lb2(s_i-s_j)-(s_i+s_j)\log\lb\frac{s_i}{s_j}\rb\rb \\
&= 0\nonumber.
\end{align}
The last equality follows from the fact that the expression in \eqref{equ:blup46} clearly changes its sign when $s_i$ and $s_j$ are exchanged. This is only possible if the value of the integral \eqref{equ:blaInt} vanishes. Rearranging the integral \eqref{equ:blaInt} gives:
\begin{equation}\label{eqintegral}
\int\limits_0^{\infty}\tr\left[\sigma A(t)XA(t)XA(t)\right]dt=\int\limits_0^{\infty}\tr\left[-\sigma A(t)XA(t)XA(t)+XA(t)XA(t)\right]dt . 
\end{equation}
Finally applying \eqref{firstorder},\eqref{eqintegral} to the formulas for the relative entropies \eqref{equ:relInt1} and \eqref{equ:relInt2} gives:
\begin{align*}
\frac{D\lb\sigma\|\sigma+\epsilon X\rb}{D\lb\sigma+\epsilon X\|\sigma\rb}=\frac{c+\mathcal{O}(\epsilon)}{c+\mathcal{O}(\epsilon)}\ra 1
\end{align*}
as $\epsilon\ra 0$. Here $c:=\int\limits_0^{\infty}\tr\left[\sigma A(t)XA(t)XA(t)\right]dt >0$ as $\sigma, A(t)>0$ for any $t\in\lbr 0,\infty\rb$ and $X\neq 0$ is Hermitian. 

\end{proof}

\bibliographystyle{IEEEtran}
\bibliography{mybibliography}

\end{document}